\documentclass[sigconf]{acmart}
\usepackage{microtype}
\usepackage{graphicx}
\usepackage{multirow}
\usepackage{multicol}
\usepackage{subcaption}
\usepackage{graphicx}
\usepackage{booktabs} 
\usepackage{bm}
\usepackage{hyperref}

\usepackage{amsmath}
\usepackage{mathtools}
\usepackage{amsthm}

\newtheorem{theorem}{Theorem}[section]
\newtheorem{proposition}[theorem]{Proposition}

\AtBeginDocument{%
  }

\setcopyright{acmlicensed}
\copyrightyear{2018}
\acmYear{2018}
\acmDOI{XXXXXXX.XXXXXXX}
\acmISBN{978-1-4503-XXXX-X/2018/06}

\begin{document}

\title{DIRECTOR: Dynamic Index-based Recommendation with Transport-Optimized Retrieval}

\author{Yuanhao Pu}
\authornote{These authors contributed equally to this research.}
\email{puyuanhao@mail.ustc.edu.cn}
\affiliation{%
  \institution{University of Science \& Technology of China}
  \city{Hefei}
  \state{Anhui}
  \country{China}
}

\author{Chenghao Zhang}
\authornotemark[1]
\email{zhangchenghao03@kuaishou.com}
\affiliation{%
  \institution{Kuaishou Technology}
  \city{Beijing}
  \country{China}
}

\author{Chao Feng}
\authornotemark[1]
\email{fengchao08@kuaishou.com}
\correspondingauthor
\affiliation{%
  \institution{Kuaishou Technology}
  \city{Beijing}
  \country{China}
}

\author{Xiang Li}
\affiliation{%
  \institution{Kuaishou Technology}
  \city{Beijing}
  \country{China}}

\author{Defu Lian}
\affiliation{%
  \institution{University of Science \& Technology of China}
  \city{Hefei}
  \state{Anhui}
  \country{China}}


\renewcommand{\shortauthors}{Pu et al.}

\begin{abstract}
Reranking is a combinatorial decision problem that aims to select and order a high-utility slate from a request-specific candidate set. A major line of generative rerankers adopts autoregressive (\textbf{AR}) models, which construct the slate one position at a time to capture inter-position dependencies. However, under practical greedy or bounded-width decoding, prefix-based search may prematurely prune globally promising permutations and incurs inherently sequential latency, restricting the effective search space under a fixed serving budget. Non-autoregressive (\textbf{NAR}) alternatives alleviate this efficiency bottleneck through position-parallel prediction, but naive position-wise factorization treats different positions too independently, leading to insufficient cross-position coordination and potentially duplicate or conflicting item selections. To retain parallel efficiency while introducing global structural coordination, we propose \textbf{D}ynamic \textbf{I}ndex-based \textbf{REC}ommendation with \textbf{T}ransport-\textbf{O}ptimized \textbf{R}etrieval (\textbf{DIRECTOR}), a transport-guided parallel reranking framework. \textbf{DIRECTOR} maps candidate items into a continuous latent space and generates request-conditioned dynamic retrieval indices for all target positions in parallel. During training, it uses entropy-regularized \textbf{OT}
to provide conflict-aware supervision; at inference,
it directly performs global hard matching on similarity
matrix, producing duplicate-free slates without iterative transport. To further align the generator with an opaque list-wise evaluator that returns only a scalar utility, we introduce a prefix-anchored credit assignment mechanism that converts the global reward into position-specific training signals. Extensive offline and online experiments demonstrate that \textbf{DIRECTOR} consistently outperforms strong reranking baselines, achieving significant improvement in large-scale industrial recommendation scenarios.

\end{abstract}



\received{20 February 2007}
\received[revised]{12 March 2009}
\received[accepted]{5 June 2009}

\maketitle

\section{Introduction}
\label{sec:intro}

Reranking is a critical stage in modern large-scale recommendation systems, where a request-specific candidate set produced by upstream retrieval and ranking modules is refined into a sequence slate. Unlike independent point-wise scoring, the quality of a reranked slate cannot generally be fully characterized by the sum of its individual item utilities. It also depends on position effects, competition and complementarity among selected items, and the overall coherence of the slate. Reranking is therefore inherently a structured combinatorial decision problem: given a candidate set of size $M$, selecting and ordering $n$ distinct items induces a search space of $\mathrm{P}(M,n)=M!/(M-n)!$ possible permutations. Even for a moderate setting with $M=50$ and $n=10$, this space already contains approximately $3.73\times 10^{16}$ valid slates.

A growing line of recent reranking research adopts a two-stage Generator-Evaluator (\textbf{G-E}) paradigm~\cite{shi2023pier,lin2023discrete,ren2024nar4rec,yang2025comprehensive}. The Generator explores the combinatorial space and proposes multiple candidate slates, while a list-wise Evaluator estimates their holistic utilities and selects the final output. A major class of learned generators is autoregressive (\textbf{AR}), constructing a slate one item at a time conditioned on previously generated items~\cite{bello2018seq2slate,feng2021grn}. Such a factorization provides a natural mechanism for modeling inter-position dependencies. However, it also couples \emph{structural dependence} with \emph{sequential computation}. Under bounded-width beam-search-like decoding~\cite{wiseman2016sequence,gong2022realtime}, a prefix that falls outside the retained beam is irrevocably pruned according to its current generator score, even though it may still be extended into a globally superior slate once later item interactions are taken into account. Moreover, each decoding step must wait for the preceding selection, resulting in inherently sequential inference~\cite{ren2024nar4rec}. Consequently, under the strict serving budget of an industrial system, only a small subset of the full permutation space, largely determined by early prefix decisions, can be effectively explored~\cite{gong2022realtime,shi2023pier}.

Non-autoregressive (\textbf{NAR}) reranking offers a promising alternative by replacing left-to-right decoding with parallel position-wise prediction or joint slate construction~\cite{jiang2019listcvae,ren2024nar4rec,xu2026omgrec}. Diffusion-based generators further relax the fixed autoregressive order, although they generally require multiple iterative denoising steps~\cite{lin2023discrete}. Nevertheless, parallel prediction alone does not guarantee a globally coordinated slate. When position-wise distributions are decoded independently, different positions may compete for the same high-probability items, producing duplicate or poorly coordinated assignments that require additional resolution. Some methods introduce dependencies through previously selected-item masking or contrastive decoding~\cite{ren2024nar4rec}, but such mechanisms reintroduce a sequential selection procedure after the parallel forward pass. More recent one-shot matching approaches construct a position-candidate allocation matrix and incorporate request-level or permutation-level modeling~\cite{xu2026omgrec}. However, two challenges remain in the Generator-Evaluator setting: how to produce sufficiently exploratory, request-conditioned position intents that adapt to a dynamically changing candidate set, and how to optimize the resulting coupled generator when the downstream Evaluator exposes only a scalar utility for the complete slate. The central question is therefore: \emph{Can we retain position-parallel generation while globally coordinating position assignments and learning from opaque list-wise feedback?}

To address this question, we propose \textbf{DIRECTOR} (\textbf{D}ynamic \textbf{I}ndex-based \textbf{REC}ommendation with \textbf{T}ransport-\textbf{O}ptimized \textbf{R}etrieval), a globally coordinated, position-parallel reranking framework. Instead of directly predicting a discrete item at each position, \textbf{DIRECTOR} jointly generates a matrix of request-conditioned dynamic retrieval indices for all target positions. Each index serves as a continuous, position-aware latent query conditioned on both the user context and the current candidate set, rather than relying solely on a static position embedding or directly representing a discrete item identifier. Candidate items are mapped into the same latent space, allowing each position index to express a flexible retrieval preference over the request-specific candidate pool. By sampling multiple latent index matrices from a conditional generative distribution, \textbf{DIRECTOR} produces multiple structured slate proposals without autoregressive rollouts.

The position-wise indices must nevertheless be globally coordinated, since a valid slate cannot assign the same item to multiple positions. We therefore introduce an entropy-regularized, capacity-constrained
optimal transport (OT) objective during training to couple
position-wise preferences and expose their competition for candidate
capacity~\cite{cuturi2013sinkhorn}. The shared transport feasible set couples all position--item assignments, requiring every target position to be filled while limiting the allocation capacity of each candidate. At inference, we bypass the soft transport solver and directly compute
a global discrete assignment from the position--candidate similarity
matrix, producing a valid duplicate-free slate. Unlike independent row-wise retrieval followed by post-hoc conflict resolution, the global hard assignment considers all position--candidate scores jointly, without introducing an \textbf{AR} decoding chain. \textbf{DIRECTOR} thus removes item-by-item sequential dependence while preserving cross-position coordination. However, learning such a coupled generator is challenging when the downstream Evaluator is opaque or non-differentiable and returns only a scalar utility for a complete slate. Broadcasting the same global reward to every position provides only coarse credit assignment. We therefore construct a validity-preserving, prefix-anchored path between the generated and baseline slates and define each position-wise advantage as the reward difference between two adjacent hybrid slates along this path. These local advantages form an exact telescoping decomposition of the global reward improvement and provide fine-grained signals for transport-relaxed reward-guided optimization.

Our main contributions are summarized as follows:
\begin{itemize}
    \item We propose \textbf{DIRECTOR}, a position-parallel reranking framework that generates request-conditioned dynamic retrieval indices, uses capacity-constrained optimal transport for conflict-aware training, performs direct global matching for duplicate-free decoding without autoregressive rollouts, and introduces prefix-anchored credit assignment to derive position-specific learning signals from opaque list-wise rewards.
    
    \item We theoretically characterize the assignment geometry and capacity-induced coupling of the transport surrogate, bound the approximation error introduced by entropy regularization, and analyze finite-proposal coverage and online inference complexity.
    
    \item Extensive offline and online experiments demonstrate that \textbf{DIRECTOR} achieves superior effectiveness and practical inference efficiency over strong \textbf{AR} and \textbf{NAR} reranking baselines in large-scale recommendation scenarios.
\end{itemize}

\begin{figure*}[t]
    \centering
    \includegraphics[width=\textwidth]{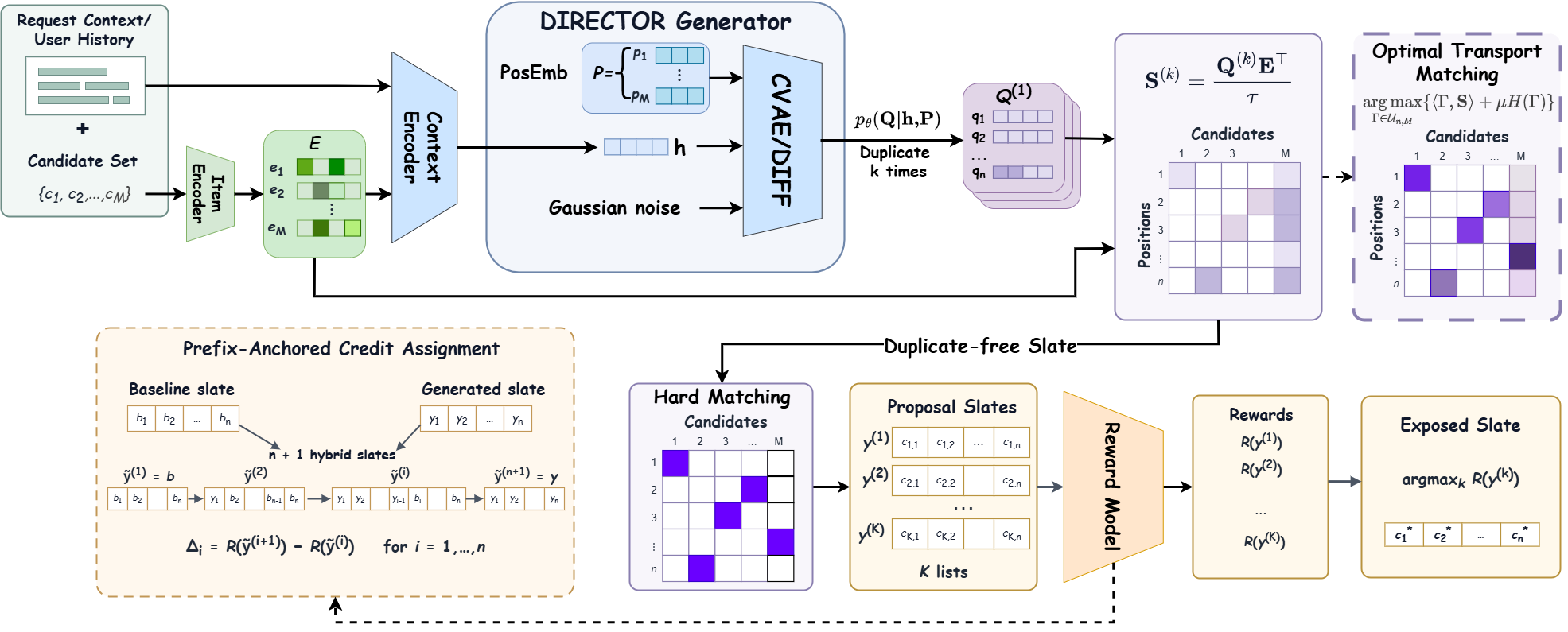}
    \caption{The framework of DIRECTOR}
    \label{fig:director}
\end{figure*}
\section{Related Work}
\label{sec:related_work}

\subsection{From Context-Aware Scoring to Generator--Evaluator Reranking}

Reranking is typically deployed near the end of a multi-stage recommender system to refine a request-specific candidate set into an ordered slate. Early ranking models score candidates independently and obtain the output by sorting their item-wise scores, overlooking the influence of positions and surrounding items. Context-aware
rerankers address this limitation by explicitly modeling intra-list interactions. DLCM~\cite{ai2018learning} propagates contextual information along the initial ranking list with a recurrent architecture. PRM~\cite{pei2019personalized} employs self-attention to capture mutual influences among candidate items, while SetRank~\cite{pang2020setrank} uses permutation-invariant set modeling for joint document ranking. Subsequent methods further incorporate personalized histories and multi-level interactions~\cite{li2022pear,xi2022multilevel}. Beyond score-based refinement, generative rerankers directly construct slates from the candidate permutation space. PRS~\cite{feng2021revisit} reformulates recommendation from a permutation perspective, and Seq2Slate~\cite{bello2018seq2slate} generates ordered slates autoregressively. More recent industrial approaches optimize generated lists under multi-objective settings~\cite{meng2025generative} or jointly incorporate causal and utility signals~\cite{zhang2026dual}. These Generator-only methods improve list construction through contextual scoring or structured generation, but do not explicitly generate and evaluate multiple slate proposals with a separate list-wise Evaluator.

A growing line of work adopts the Generator--Evaluator (\textbf{G-E})
paradigm, in which a Generator explores the permutation space and an
Evaluator estimates the utilities of the proposed slates.
Generator-and-Critic~\cite{wei2020generator} and GRN~\cite{feng2021grn}
optimize sequential slate generators using a learned Critic or Evaluator.
JDRec~\cite{zhao2024jdrec} develops an actor--critic framework with an
emphasis on practical deployment and policy bootstrapping.
PIER~\cite{shi2023pier} jointly improves permutation generation and
list-wise evaluation, while recent multi-generator methods enlarge the
proposal space by combining complementary generation policies
~\cite{yang2025comprehensive}.

The G-E decomposition separates slate exploration from utility
estimation, but its effectiveness remains constrained by the Generator. Moreover, when the Evaluator is opaque and
provides only a scalar utility for each complete slate, obtaining
fine-grained supervision for individual generation decisions remains challenging. Our work focuses on generator-side problems while treating the Evaluator as a list-wise reward model.

\subsection{Autoregressive and Non-Autoregressive Slate Generation}

A major line of generative reranking methods adopts an
autoregressive (\textbf{AR}) factorization, selecting each item
conditioned on the previously generated prefix. Seq2Slate
~\cite{bello2018seq2slate} uses a pointer network to generate the
slate item by item, while GFN4Rec~\cite{liu2023generative} models slate
construction as a sequential trajectory and learns to sample diverse
high-reward lists. Prefix conditioning provides a natural mechanism
for modeling dependencies among selected items and permits previously
selected candidates to be masked. However, it also introduces an
item-by-item dependency chain, and bounded-width beam decoding can
only explore continuations of the prefixes retained at intermediate
steps.

Non-autoregressive (\textbf{NAR}) approaches avoid fixed
left-to-right item generation through several different mechanisms. List-CVAE~\cite{jiang2019listcvae} models a conditional distribution over
complete slates using latent-variable generation. NAR4Rec
~\cite{ren2024nar4rec} predicts target positions in parallel and
introduces matching, sequence-level unlikelihood training, and
contrastive decoding to improve intra-slate consistency, although its
final decoding procedure still depends on previously selected items.
DCDR~\cite{lin2023discrete} generates permutations through iterative
discrete diffusion, removing the fixed autoregressive order at the
cost of multiple denoising steps. Most closely related to our work,
OMGRec~\cite{xu2026omgrec} directly constructs a request-aware
position--candidate allocation matrix and applies one-time matching
with permutation-level modeling.

In contrast, \textbf{DIRECTOR} first samples continuous dynamic
retrieval indices rather than directly predicting candidate
allocations, enabling multiple slate proposals through latent-space
sampling. It then uses capacity-constrained soft transport to couple position-wise assignments during training, and directly applies global hard matching for valid decoding, and further optimizes the coupled generator from opaque complete-slate rewards through prefix-anchored pathwise credit assignment.
\section{Preliminaries}
In this section, we formalize list-wise reranking under the
Generator--Evaluator (\textbf{G-E}) paradigm. We then review
\textbf{AR} generation and sequence-level reward-guided
optimization, which motivate the globally coordinated parallel
policy introduced in Section~\ref{sec:method}.

\subsection{Problem Formulation}
\label{subsec:problem_formulation}

Let $\mathcal{V}$ denote the entire item corpus. Given a request,
let $\bm{x}$ denote the complete request context, including a
historical interaction sequence
$\bm{x}^{\mathrm{hist}}=(x_1,x_2,\ldots,x_T)$ with
$x_t\in\mathcal{V}$, together with additional user- and
request-level features. A preceding retrieval or ranking stage
provides a request-specific candidate set
$\mathcal{C}=\{c_1,c_2,\ldots,c_M\}\subseteq\mathcal{V}$,
where $M=|\mathcal{C}|$.

Given a target slate size $n$ with $1\le n\le M$, the reranking
task selects and orders $n$ distinct items from $\mathcal{C}$ to
form a slate $\bm{y}=(y_1,y_2,\ldots,y_n)$. Let
$[n]=\{1,\ldots,n\}$. The feasible slate space is
\begin{equation}
    \mathcal{S}(\mathcal{C},n)
    =
    \left\{
        \bm{y}\in\mathcal{C}^{n}
        \,\middle|\,
        y_i\neq y_j,\ 
        \forall i,j\in[n],\ i\neq j
    \right\},
    \label{eq:feasible_slate_space}
\end{equation}
whose cardinality is $ \left|\mathcal{S}(\mathcal{C},n)\right|=\mathrm{P}(M,n)=\frac{M!}{(M-n)!}$.
Let $U(\bm{y}\mid\bm{x},\mathcal{C})\in\mathbb{R}$ denote the
underlying list-wise utility. An ideal reranking decision satisfies
\begin{equation}
    \bm{y}^{\star}
    \in
    \operatorname*{arg\,max}_{\bm{y}\in
    \mathcal{S}(\mathcal{C},n)}
    U(\bm{y}\mid\bm{x},\mathcal{C}).
    \label{eq:optimal_slate}
\end{equation}
Therefore, an exact optimization may require searching over $\mathrm{P}(M,n)$ feasible slates, which is computationally prohibitive.

\subsection{The Generator-Evaluator Paradigm}
\label{subsec:generator_evaluator}

The G-E paradigm approximates Eq.~\eqref{eq:optimal_slate} by
restricting the decision to a finite collection of generated proposals
and using an Evaluator as a proxy for the underlying list-wise utility.
It thereby separates candidate-slate exploration from utility
estimation.

\textbf{Generator.}
A valid Generator induces a stochastic policy
$\pi_{\theta}(\bm{y}\mid\bm{x},\mathcal{C})$ over
$\mathcal{S}(\mathcal{C},n)$. Given a request, it produces a
collection of $K$ proposal slates
$\mathcal{Y}_{1:K}=\left(\bm{y}^{(1)},\bm{y}^{(2)},\ldots,\bm{y}^{(K)}\right)$. Under stochastic generation,
$\bm{y}^{(k)}\sim
\pi_{\theta}(\cdot\mid\bm{x},\mathcal{C})$; alternatively,
$\mathcal{Y}_{1:K}$ may be obtained through an approximate search
procedure such as beam search.

A conventional \textbf{AR} Generator factorizes the joint policy according to the generation order:
\begin{equation}
    \pi_{\theta}^{\mathrm{AR}}
    (\bm{y}\mid\bm{x},\mathcal{C})
    =
    \prod_{t=1}^{n}
    \pi_{\theta,t}
    \left(
        y_t
        \mid
        \bm{y}_{<t},\bm{x},\mathcal{C}
    \right),
    \label{eq:ar_generation}
\end{equation}
where $\bm{y}_{<t}=(y_1,\ldots,y_{t-1})$ denotes the generated
prefix. For a valid slate policy, each conditional distribution is
supported on the remaining candidates
$\mathcal{C}\setminus\{y_1,\ldots,y_{t-1}\}$.
Conditioning on $\bm{y}_{<t}$ provides a natural mechanism for
modeling dependencies among selected items and allows previously
selected candidates to be masked. However, at inference time, this
factorization imposes a sequential decoding chain in which the
decision at position $t$ cannot be made before the preceding
positions have been instantiated.

A naive \textbf{NAR} alternative predicts all positions simultaneously:
\begin{equation}
    \widetilde{\pi}_{\theta}^{\mathrm{NAR}}
    (\bm{y}\mid\bm{x},\mathcal{C})
    =
    \prod_{t=1}^{n}
    p_{\theta,t}
    (y_t\mid\bm{x},\mathcal{C}),
    \quad
    \bm{y}\in\mathcal{C}^{n}.
    \label{eq:factorized_nar}
\end{equation}
While Eq.~\eqref{eq:factorized_nar} permits position-parallel
prediction, its product distribution is defined over $\mathcal{C}^{n}$ rather than a valid policy over $\mathcal{S}(\mathcal{C},n)$. Consequently, different positions
may assign the same candidate multiple times. Even when duplicates
are resolved post hoc, the factorized decisions do not jointly
account for competition over candidate capacity. This mismatch
motivates a structured parallel policy that generates position-wise
intents simultaneously while coordinating their discrete assignments
globally.

\textbf{Evaluator.}
The Evaluator is a list-wise reward model
$R_{\phi}(\bm{y}\mid\bm{x},\mathcal{C})\in\mathbb{R}$ that serves
as a proxy for the underlying utility $U$. It may be a learned
neural model, a composite business objective, or an opaque industrial
service whose internal parameters and gradients are unavailable to
the Generator. Given the proposal collection
$\mathcal{Y}_{1:K}$, the selected output is
\begin{equation}
    \hat{k}
    \in
    \operatorname*{arg\,max}_{k\in[K]}
    R_{\phi}
    \left(
        \bm{y}^{(k)}
        \mid
        \bm{x},\mathcal{C}
    \right),
    \quad
    \hat{\bm{y}}
    =
    \bm{y}^{(\hat{k})}.
    \label{eq:evaluator_argmax}
\end{equation}
The resulting quality depends on both the Generator's ability to
include high-utility slates in the proposal collection and the
Evaluator's fidelity in ranking those proposals according to $U$.

\subsection{Reward-Guided Generator Optimization}
\label{subsec:reward_guided_optimization}

Although online serving follows the best-of-$K$ decision in Eq.~\eqref{eq:evaluator_argmax}, a common surrogate for improving the Generator is to maximize the expected reward of a sampled slate:
\begin{equation}
    J(\theta)
    =
    \mathbb{E}_{\substack{
        (\bm{x},\mathcal{C})\sim\mathcal{D}\\
        \bm{y}\sim\pi_{\theta}(\cdot\mid\bm{x},\mathcal{C})
    }}
    \left[
        R_{\phi}(\bm{y}\mid\bm{x},\mathcal{C})
    \right],
    \label{eq:rl_objective}
\end{equation}
where $\mathcal{D}$ denotes the distribution over
request--candidate-set pairs. With the Evaluator fixed, a
score-function estimator gives
\begin{equation}
    \nabla_{\theta}J(\theta)
    =
    \mathbb{E}
    \left[
        \big(
            R_{\phi}(\bm{y}\mid\bm{x},\mathcal{C})
            -
            b(\bm{x},\mathcal{C})
        \big)
        \nabla_{\theta}
        \log\pi_{\theta}
        (\bm{y}\mid\bm{x},\mathcal{C})
    \right],
    \label{eq:general_pg}
\end{equation}
where $b(\bm{x},\mathcal{C})$ is an action-independent baseline. For \textbf{AR} Generator,
\begin{equation}
    \nabla_{\theta}
    \log\pi_{\theta}^{\mathrm{AR}}
    (\bm{y}\mid\bm{x},\mathcal{C})
    =
    \sum_{t=1}^{n}
    \nabla_{\theta}
    \log\pi_{\theta,t}
    \left(
        y_t
        \mid
        \bm{y}_{<t},\bm{x},\mathcal{C}
    \right).
    \label{eq:ar_score_decomposition}
\end{equation}
Thus, when only a complete-slate reward and a request-level baseline are available, the sequence-level
advantage is applied to every position-wise score term. Such supervision provides no direct indication of which selections improve or degrade the utility, particularly when the Evaluator captures cross-position interactions.

These observations motivate \textbf{DIRECTOR}, which replaces
item-by-item \textbf{AR} decoding with parallel dynamic-index generation and globally coordinated hard matching. It further derives position-specific pathwise credits from reward differences along a validity-preserving hybrid path, as detailed in Section~\ref{sec:method}.

\section{Methodology}
\label{sec:method}

We present \textbf{DIRECTOR}, a position-parallel reranking framework
that separates continuous intent generation from discrete slate
construction. As illustrated in Figure~\ref{fig:director},
\textbf{DIRECTOR} jointly generates request-conditioned dynamic
retrieval indices for all target positions. It uses capacity-constrained soft transport to coordinate position-wise preferences during training and direct global hard matching to construct valid slates at inference. Finally, a fixed list-wise Evaluator, which may be opaque or
non-differentiable, guides the Generator through prefix-anchored
pathwise credits.

Given the request context $\bm{x}$ and candidate set
$\mathcal{C}=\{c_1,\ldots,c_M\}$, let $\bm{v}_j$ denote the
available features associated with candidate $c_j$. An item encoder
maps each candidate into a shared retrieval space:
\begin{equation*}
    \bm{e}_j
    =
    f_{\mathrm{item}}(c_j,\bm{v}_j)
    \in\mathbb{R}^{d},
    \quad
    \bm{E}
    =
    [\bm{e}_1,\ldots,\bm{e}_M]^\top
    \in\mathbb{R}^{M\times d}.
\end{equation*}
A context encoder jointly summarizes the request context and the current candidate pool with $\bm{h}=f_{\mathrm{ctx}}(\bm{x},\bm{E})\in\mathbb{R}^{d_h}$. The candidate embeddings are computed once per request and reused across all target positions and sampled proposals. Unless otherwise specified, the parameters of $f_{\mathrm{item}}$ and $f_{\mathrm{ctx}}$ are included in the overall Generator parameters $\theta$.

\subsection{Dynamic Index Generation}
\label{subsec:query_generation}

Instead of predicting a discrete item for each output position, \textbf{DIRECTOR} generates a matrix of continuous retrieval indices,
\begin{equation*}
    \bm{Q}
    =
    [\bm{q}_1,\ldots,\bm{q}_n]^\top
    \sim
    p_{\theta}(\bm{Q}\mid\bm{h},\bm{P}),
    \qquad
    \bm{Q}\in\mathbb{R}^{n\times d},
\end{equation*}
where $\bm{P}=[\bm{p}_1,\ldots,\bm{p}_n]^\top$ contains learnable position embeddings. Each row $\bm{q}_i$ represents the retrieval intent of position $i$. All indices are generated jointly rather than conditioned on a previously selected item. Repeatedly sampling $\bm{Q}$ therefore produces multiple slate proposals without autoregressive rollouts.

\textbf{Training target.} For an observed slate $\bm{y}=(y_1,\ldots,y_n)\in\mathcal{S}(\mathcal{C},n)$, we define the target index at position $i$ by combining the item embedding with the corresponding position embedding:
\begin{equation*}
    \overline{\bm{Q}}
    =
    \left[
        \bm{e}(y_1)+\bm{p}_1,
        \ldots,
        \bm{e}(y_n)+\bm{p}_n
    \right]^\top.
\end{equation*}
We consider two conditional generative models for learning the
distribution of such index matrices.

\textbf{Conditional VAE.} During training, the posterior
$q_{\psi}(\bm{z}\mid\overline{\bm{Q}},\bm{h})$ infers a latent
variable from the target indices, while
$p_{\theta}(\bm{z}\mid\bm{h})$ serves as the inference prior. The
decoder jointly reconstructs all position indices with:
\begin{equation}
    \mathcal{L}_{\mathrm{CVAE}}
    =
    \mathbb{E}_{\bm{z}\sim
    q_{\psi}}
    \left[
        \left\|
            g_{\theta}(\bm{z},\bm{h},\bm{P})
            -
            \overline{\bm{Q}}
        \right\|_F^2
    \right]
    +
    \beta D_{\mathrm{KL}}
    \left(
        q_{\psi}
        \,\Vert\,
        p_{\theta}
    \right).
    \label{eq:cvae_loss}
\end{equation}
At inference, different samples yield different index matrices.

\textbf{Conditional diffusion.}
Alternatively, we regard
$\bm{Q}_0=\overline{\bm{Q}}$ as the clean index matrix and corrupt
it through
\begin{equation*}
    \bm{Q}_t
    =
    \sqrt{\overline{\alpha}_t}\bm{Q}_0
    +
    \sqrt{1-\overline{\alpha}_t}\bm{\epsilon},
    \quad
    \bm{\epsilon}\sim\mathcal{N}(\bm{0},\bm{I}),
\end{equation*}
where $\overline{\alpha}_t$ is the cumulative noise schedule. The
conditional denoiser is optimized by
\begin{equation}
    \mathcal{L}_{\mathrm{DIFF}}
    =
    \mathbb{E}_{t,\bm{\epsilon}}
    \left[
        \left\|
            \bm{\epsilon}
            -
            \epsilon_{\theta}
            (\bm{Q}_t,t,\bm{h},\bm{P})
        \right\|_F^2
    \right].
    \label{eq:diff_loss}
\end{equation}
At inference, the reverse process starts from Gaussian noise and
jointly updates all position indices at every denoising step.

Both implementations avoid an item-by-item \textbf{AR} chain.
CVAE generates the index matrix in one forward pass,
whereas diffusion performs few denoising steps that remain
parallel across positions. We use $\mathcal{L}_{\mathrm{Index}}$ to denote the selected index-generation
objective.

\subsection{Transport-Guided Parallel Retrieval}
\label{subsec:parallel_retrieval}

Given the dynamic indices $\bm{Q}$ and candidate embeddings $\bm{E}$,
we compute the position--candidate similarity matrix
\begin{equation*}
    \bm{S}
    =
    \frac{\bm{Q}\bm{E}^{\top}}{\tau}
    \in\mathbb{R}^{n\times M},
    \quad
    S_{ij}
    =
    \frac{\bm{q}_i^{\top}\bm{e}_j}{\tau},
\end{equation*}
where $S_{ij}$ measures how candidate $c_j$ matches output position $i$, and $\tau>0$ is a temperature. Selecting the highest-scoring candidate independently at each position may produce duplicate items. Therefore, we instead construct the slate through a global assignment:
\begin{equation}
\begin{aligned}
    \bm{A}^{\star}
    &=
    \operatorname*{arg\,max}_{\bm{A}\in\mathcal{A}_{n,M}}
    \langle\bm{A},\bm{S}\rangle,
    \\ 
    \mathcal{A}_{n,M}
    &=
    \left\{
        \bm{A}\in\{0,1\}^{n\times M}\middle|
        \bm{A}\bm{1}_{M}=\bm{1}_{n},
        \bm{A}^{\top}\bm{1}_{n}\leq\bm{1}_{M}
    \right\}.
    \label{eq:hard_assignment}
\end{aligned}
\end{equation}
We solve Eq.~\eqref{eq:hard_assignment} using a rectangular shortest-augmenting-path algorithm
~\cite{crouse2016implementing}. If $A_{ij}^{\star}=1$, candidate $c_j$ is placed at position $i$. The row constraints fill every output position, while the column constraints ensure that each candidate is selected at most once. The resulting slate is therefore complete and duplicate-free. 

\textbf{Transport-guided training.}
The hard assignment in Eq.~\eqref{eq:hard_assignment} is suitable for
slate construction but does not provide smooth gradients. During
training, we additionally compute an entropy-regularized soft assignment.
Let
\[
    \mathcal{U}_{n,M}
    =
    \left\{
        \bm{\Gamma}\in\mathbb{R}_{+}^{n\times M}
        \,\middle|\,
        \bm{\Gamma}\bm{1}_{M}=\bm{1}_{n},\;
        \bm{\Gamma}^{\top}\bm{1}_{n}\leq\bm{1}_{M}
    \right\}.
\]
The soft transport plan is
\begin{equation}
    \bm{\Gamma}_{\mu}^{\star}
    =
    \operatorname*{arg\,max}_{\bm{\Gamma}\in\mathcal{U}_{n,M}}
    \left\{
        \langle\bm{\Gamma},\bm{S}\rangle
        +
        \mu H(\bm{\Gamma})
    \right\},
    \ 
    H(\bm{\Gamma})
    =
    -
    \sum_{i=1}^{n}
    \sum_{j=1}^{M}
    \Gamma_{ij}\log\Gamma_{ij},
    \label{eq:ot}
\end{equation}
where $\mu>0$ controls the entropy regularization. The shared column constraints couple the position-wise assignments. When several positions prefer the same candidate, their transport masses are adjusted jointly rather than normalized independently. Following~\cite{cuturi2013sinkhorn}, we initialize
$\bm{K}=\exp(\bm{S}/\mu)$ and solve Eq.~\eqref{eq:ot} by alternating row normalization and column-capacity projection. Since the column constraint is an inequality, we adopt the Bregman-Dykstra Sinkhorn~\cite{benamou2015iterative}. At inference, we directly apply the hard assignment in Eq.~\eqref{eq:hard_assignment} to $\bm{S}$, avoiding iterative transport computation under strict industrial latency constraints.

To generate $K$ proposals, we independently sample
\[
    \bm{Q}^{(k)}
    \overset{\mathrm{i.i.d.}}{\sim}
    p_{\theta}(\cdot\mid\bm{h},\bm{P}),
    \quad
    k=1,\ldots,K,
\]
and solve Eq.~\eqref{eq:hard_assignment} for each corresponding
similarity matrix $\bm{S}^{(k)}$. This produces
$\{\bm{y}^{(1)},\ldots,\bm{y}^{(K)}\}$ for the downstream Evaluator.
The similarity computation and the $K$ independent matching problems
can be executed in parallel.

\subsection{Prefix-Anchored Credit and Reward-Guided Optimization}
\label{subsec:reward_optimization}

The hard decoder produces valid, duplicate-free slates, but it does not directly optimize their list-wise utility. We therefore use a fixed Evaluator $R_{\phi}(\bm{y}\mid\bm{x},\mathcal{C})$ to guide the Generator. Since the Evaluator returns only one scalar, applying the same reward to every position provides little information about individual assignments.

To obtain position-specific feedback, we build a valid path from a baseline slate $\bm b$ to the generated slate $\bm y$. Here, $\bm b\in\mathcal S(\mathcal C,n)$ denotes the logged exposure slate associated with the request. Starting from $\widetilde{\bm{y}}^{(1)}=\bm{b}$, step $i$ places $y_i$ at position
$i$. If $y_i$ already appears later in the current slate, the two items
are swapped; otherwise, the item at position $i$ is replaced. Each step
keeps the slate duplicate-free, and the final state is
$\widetilde{\bm{y}}^{(n+1)}=\bm{y}$. We define the credit for position $i$ as the reward change between two
adjacent slates:
\begin{equation*}
    \Delta_i
    =
    R_{\phi}
    \left(
        \widetilde{\bm{y}}^{(i+1)}
        \mid\bm{x},\mathcal{C}
    \right)
    -
    R_{\phi}
    \left(
        \widetilde{\bm{y}}^{(i)}
        \mid\bm{x},\mathcal{C}
    \right).
\end{equation*}
These credits exactly decompose the reward improvement:
\begin{equation*}
    \sum_{i=1}^{n}\Delta_i
    =
    R_{\phi}(\bm{y}\mid\bm{x},\mathcal{C})
    -
    R_{\phi}(\bm{b}\mid\bm{x},\mathcal{C}).
\end{equation*}
Thus, the method requires only scalar Evaluator outputs, and all
$n+1$ intermediate slates can be scored in one batch. Related ideas
appear in path-based attribution and multi-agent credit assignment
~\cite{sundararajan2017axiomatic,wolpert2001optimal}; here, the path is
designed specifically to preserve slate validity.

Let $a_i$ denote the candidate assigned to position $i$ by the hard
decoder, i.e., $A_{i,a_i}^{\star}=1$. During training, we use the soft
transport plan $\bm{\Gamma}_{\mu}^{\star}$ to pass gradients:
\begin{equation}
    \mathcal{L}_{\mathrm{CA}}
    =
    -
    \mathbb{E}_{\bm{Q}\sim
    p_{\theta}(\cdot\mid\bm{h},\bm{P})}
    \left[
        \sum_{i=1}^{n}
        \operatorname{sg}(\Delta_i)
        \log
        \left(
            \Gamma_{\mu,i,a_i}^{\star}+\epsilon
        \right)
    \right].
    \label{eq:credit_assignment_loss}
\end{equation}
A positive $\Delta_i$ reinforces the selected assignment, whereas a
negative one suppresses it. For warm-starting, let $\bm{B}\in\mathcal{A}_{n,M}$ denote the
assignment matrix of the baseline slate. We use
\begin{equation}
    \mathcal{L}_{\mathrm{Match}}
    =
    -
    \left\langle
        \bm{B},
        \log(\bm{\Gamma}_{\mu}^{\star}+\epsilon)
    \right\rangle.
    \label{eq:matching_loss}
\end{equation}

\subsection{Overall Training Objective}
\label{subsec:overall_loss}

The complete objective is
\begin{equation}
    \mathcal{L}_{\mathrm{Total}}
    =
    \mathcal{L}_{\mathrm{Index}}
    +
    \alpha\mathcal{L}_{\mathrm{Match}}
    +
    \lambda\mathcal{L}_{\mathrm{CA}},
    \label{eq:total_loss}
\end{equation}
where $\alpha$ and $\lambda$ control supervised matching and
Evaluator-guided learning, respectively. We first warm-start the model
with
$\mathcal{L}_{\mathrm{Index}}+\alpha\mathcal{L}_{\mathrm{Match}}$,
and then optimize the full objective.

\section{Theoretical Results}
\label{sec:main_results}

We analyze the assignment geometry of hard matching and soft
transport, the approximation gap induced by entropy regularization, finite-proposal coverage, and online inference complexity. We fix a request $(\bm{x},\mathcal{C})$ and omit it whenever no ambiguity arises.

\subsection{Hard Matching and Soft Transport}
\label{subsec:transport_geometry}

\textbf{DIRECTOR} uses hard matching to construct slates at inference and soft transport to provide differentiable guidance during training. We first
show that hard matching exactly represents the valid space.

\begin{theorem}[Slate-Assignment]
\label{thm:assignment_integrality}
The binary assignment set $\mathcal{A}_{n,M}$ is in bijection with the feasible slate space $\mathcal{S}(\mathcal{C},n)$. Moreover,
\begin{equation*}
    \mathcal{U}_{n,M}
    =
    \operatorname{conv}
    \left(
        \mathcal{A}_{n,M}
    \right).
\end{equation*}
Therefore, for any score matrix
$\bm{S}\in\mathbb{R}^{n\times M}$,
\begin{equation*}
    \max_{\bm{A}\in\mathcal{A}_{n,M}}
    \langle\bm{A},\bm{S}\rangle
    =
    \max_{\bm{\Gamma}\in\mathcal{U}_{n,M}}
    \langle\bm{\Gamma},\bm{S}\rangle,
\end{equation*}
and the relaxed problem admits an integral optimizer $\bm{A}^{\star}\in\mathcal{A}_{n,M}$ corresponding to a valid slate.
\end{theorem}

The proof is given in
Appendix~\ref{app:proof_assignment_integrality}.
Theorem~\ref{thm:assignment_integrality} shows that every valid slate
corresponds to exactly one binary assignment, and the continuous
relaxation has no integrality gap. This justifies the global hard
matching used at inference.

During training, \textbf{DIRECTOR} replaces the hard assignment with an entropy-regularized soft transport plan. The next result shows how its capacity constraints couple different positions.

\begin{theorem}[Uniqueness and Capacity-Induced Coupling]
\label{thm:ot_coupling}
For every $\mu>0$, the problem in
Eq.~\eqref{eq:ot} admits a unique optimizer
$\bm{\Gamma}_{\mu}^{\star}$ satisfying
$\Gamma_{\mu,ij}^{\star}>0$ for all $i\in[n]$ and $j\in[M]$.
Moreover, there exist dual variables
$\alpha_i\in\mathbb{R}$ for the row constraints and
$\beta_j\geq0$ for the candidate-capacity constraints such that
\begin{equation*}
    \Gamma_{\mu,ij}^{\star}
    =
    \exp
    \left(
        \frac{S_{ij}-\alpha_i-\beta_j}{\mu}-1
    \right),
    \quad
    \beta_j
    \left(
        1-\sum_{i=1}^{n}\Gamma_{\mu,ij}^{\star}
    \right)
    =
    0.
\end{equation*}

Let
\begin{equation*}
    \overline{\Gamma}_{ij}
    =
    \frac{\exp(S_{ij}/\mu)}
    {\sum_{\ell=1}^{M}\exp(S_{i\ell}/\mu)}
\end{equation*}
denote the optimizer obtained by normalizing each row independently.
If $\overline{\bm{\Gamma}}\in\mathcal{U}_{n,M}$, then
$\bm{\Gamma}_{\mu}^{\star}=\overline{\bm{\Gamma}}$.
Otherwise,
$\bm{\Gamma}_{\mu}^{\star}\neq\overline{\bm{\Gamma}}$ and
$\beta_j>0$ for at least one candidate $j$.
\end{theorem}

The proof is given in Appendix~\ref{app:proof_ot_coupling}.
The multiplier $\beta_j$ reflects the competition for candidate
$c_j$. If independent predictions satisfy all capacity constraints,
soft transport leaves them unchanged. Otherwise, it jointly adjusts
the conflicting positions. Thus, hard matching guarantees valid
inference, while soft transport provides conflict-aware training
signals.

\subsection{Approximation Behaviors}
\label{subsec:entropic_approximation}
We now bound the gap between the training surrogate defined in Eq.~\eqref{eq:ot} and the hard inference objective $\operatorname*{arg\,max}_{\bm A\in\mathcal A_{n,M}}\langle \bm A,\bm S\rangle$. By Theorem~\ref{thm:assignment_integrality}, the optimal hard-assignment
score can be written as
\[
    V_{\mathrm{hard}}(\bm S)
    :=
    \max_{\bm A\in\mathcal A_{n,M}}
    \langle\bm A,\bm S\rangle
    =
    \max_{\bm\Gamma\in\mathcal U_{n,M}}
    \langle\bm\Gamma,\bm S\rangle.
\]
\begin{theorem}[Surrogate Gap]
\label{thm:entropic_approximation}
For any $\bm S\in\mathbb R^{n\times M}$ and $\mu>0$,
\begin{equation}
    0
    \leq
    V_{\mathrm{hard}}(\bm S)
    -
    \left\langle
        \bm\Gamma_\mu^\star,\bm S
    \right\rangle
    \leq
    \mu n\log M.
    \label{eq:entropic_approximation_bound}
\end{equation}
\end{theorem}

The proof is given in
Appendix~\ref{app:proof_entropic_approximation}.
Theorem~\ref{thm:entropic_approximation} shows that the soft transport
plan used for training loses at most $\mu\log M$ in similarity score
per output position compared with the exact hard matching used at
inference. Hence, $\mu$ controls a direct trade-off: a larger value
provides smoother training signals, whereas a smaller value more
closely matches the inference objective. In particular, the score gap
vanishes as $\mu\to0$.

\subsection{Finite-Proposal Coverage}
\label{subsec:exploration_capacity}

Many \textbf{AR} rerankers generate a slate with beam decoding, under which a complete slate can be reached only if each of its prefixes remains in the beam. Once a prefix is pruned, all slates extending that prefix become unreachable in the current decoding run. \textbf{NAR} methods like \textbf{DIRECTOR} avoids this prefix-truncation mechanism. 

Let $\mathsf D(\bm Q,\bm E)$ denote the hard-assignment decoder in
Eq.~\eqref{eq:hard_assignment}, with a fixed tie-breaking rule. For any
target set
$\mathcal T\subseteq\mathcal S(\mathcal C,n)$, define its
single-sample probability as
\[
    p_{\mathcal T}
    =
    \mathbb P_{\bm Q\sim p_\theta(\cdot\mid\bm h,\bm P)}
    \left[
        \mathsf D(\bm Q,\bm E)\in\mathcal T
    \right].
\]

\begin{proposition}[Finite-Proposal Coverage]
\label{prop:multi_sample_discovery}
For $K$ independently sampled proposals,
\begin{equation}
    \mathbb P
    \left[
        \exists k\in[K]:
        \bm y^{(k)}\in\mathcal T
    \right]
    =
    1-(1-p_{\mathcal T})^K.
    \label{eq:exact_discovery_probability}
\end{equation}
\end{proposition}

The proof is given in
Appendix~\ref{app:proof_multi_sample_discovery}.
Thus, any target region with positive single-sample probability is reached with increasing probability as $K$ grows. Unlike beam search, this increase does not depend on the survival of intermediate prefixes.

Appendix~\ref{app:local_decode_stability} provides a sufficient
condition for $p_{\mathcal T}>0$. If a slate is the unique hard-matching
solution at some index matrix with a positive assignment margin, it
remains the decoder output in a neighborhood of that matrix. Assigning
positive probability to this neighborhood therefore gives a positive generation probability.

\subsection{Online Inference Complexity}

Online latency is critical for industrial deployment. We compare the generator-side cost of producing $K$ slate proposals, using beam width $B=K$ for the AR Generator. All methods use an $L$-layer Transformer with hidden dimension $d_g$.

\begin{table}[t]
\centering
\caption{Online complexity for generating $K$ proposals with
Transformer backbones, where the AR beam width is $B=K$.}
\label{tab:inference_complexity}
\setlength{\tabcolsep}{3.2pt}
\renewcommand{\arraystretch}{1.10}
\resizebox{\columnwidth}{!}{
\begin{tabular}{lccc}
\toprule
 & AR & DIRECTOR-CVAE & DIRECTOR-DIFF \\
\midrule
Transformer calls
& $n$ sequential
& $1$ batched
& $T_{\mathrm{diff}}$ batched
\\

Proposal generation
&
$\mathcal O\!\left(
KL(n^2d_g^2+n^3d_g)
\right)$
&
$\mathcal O\!\left(
KL(nd_g^2+n^2d_g)
\right)$
&
$\mathcal O\!\left(
KT_{\mathrm{diff}}L(nd_g^2+n^2d_g)
\right)$
\\

Candidate scoring
&
$\mathcal O(KnMd)$
&
$\mathcal O(KnMd)$
&
$\mathcal O(KnMd)$
\\

Slate construction
&
$\mathcal O\!\left(
nKM\log(KM)
\right)$
&
$\mathcal O(Kn^2M)$
&
$\mathcal O(Kn^2M)$
\\
\bottomrule
\end{tabular}
}
\end{table}

At decoding step $t$, AR evaluates the $K$ surviving prefixes and scores their expansions over $M$ candidates. Repeatedly processing the growing prefixes gives $\mathcal O(KL(n^2d_g^2+n^3d_g))$ Transformer computation, while beam expansion and pruning require $\mathcal O(nKM\log(KM))$ operations. \textbf{DIRECTOR} instead generates the $K$ complete index matrices in a batch and solves $K$ independent rectangular problems.

The discrete decoding terms are lower-order than similarity scoring
in typical reranking settings. Their ratios to
$\mathcal O(KnMd)$ are
$\mathcal O(\log(KM)/d)$ for AR beam pruning and
$\mathcal O(n/d)$ for \textbf{DIRECTOR} hard matching. Since the target slate
is short and $n\ll d$, global matching introduces limited overhead.
Even with KV caching, AR still requires $n$ prefix-dependent
Transformer calls and repeated beam updates, whereas \textbf{DIRECTOR} jointly
updates all positions. Avoiding this item-by-item dependency chain is
a core serving advantage of \textbf{DIRECTOR}-style \textbf{NAR} generation.
\begin{table*}[t]
    \centering
    \caption{Statistics of the datasets used in the experiments.}
    \label{tab:dataset_stats}
    \begin{tabular}{lccccc}
        \toprule
        Dataset & Domain & \# Requests & \# Items & Candidate Pool Size & Target Length \\
        \midrule
        ML-1M & RecSys & 161,646 & 3,043 & 50 & 6 \\
        Amazon-Books & RecSys & 309,917 & 38,121 & 50 & 6 \\
        RecFlow & RecSys & 3,308,233 & 14,181,768 & 120 & 6 \\
        \bottomrule
    \end{tabular}
\end{table*}

\begin{table*}[t]
\centering
\caption{Performance comparison on three datasets.
N@6, P@6, and R@6 denote NDCG@6, Precision@6, and Recall@6, respectively.
The best result is highlighted in bold, and the strongest non-DIRECTOR baseline is underlined.
Improvements are calculated relative to the strongest baseline for each metric.}
\label{tab:main_results}
\setlength{\tabcolsep}{3.2pt}
\resizebox{\textwidth}{!}{
\begin{tabular}{ll*{3}{cccc}}
\toprule
\multirow{2}{*}{Category}
& \multirow{2}{*}{Model}
& \multicolumn{4}{c}{ML-1M}
& \multicolumn{4}{c}{Amazon-Books}
& \multicolumn{4}{c}{RecFlow} \\
\cmidrule(lr){3-6}
\cmidrule(lr){7-10}
\cmidrule(lr){11-14}
& & N@6 & P@6 & R@6 & F1@6
& N@6 & P@6 & R@6 & F1@6
& N@6 & P@6 & R@6 & F1@6 \\
\midrule

\multirow{6}{*}{\shortstack{Generator-\\Only}}
& DNN
& 0.5950 & 0.4539 & 0.5542 & 0.4876
& 0.6448 & 0.5072 & 0.6125 & 0.5472
& 0.1584 & 0.0793 & 0.2069 & 0.1084 \\

& DCN
& 0.5981 & 0.4561 & 0.5573 & 0.4901
& 0.6683 & 0.5298 & 0.6461 & 0.5701
& 0.1597 & 0.0795 & 0.2083 & 0.1088 \\

& Seq2Slate
& 0.6222 & 0.4867 & 0.5927 & 0.5225
& 0.6952 & 0.5654 & 0.6871 & 0.6078
& 0.1693 & 0.0821 & 0.2134 & 0.1130 \\

& DLCM
& 0.6061 & 0.4643 & 0.5667 & 0.4988
& 0.6597 & 0.5242 & 0.6396 & 0.5641
& 0.1747 & 0.0861 & 0.2240 & 0.1169 \\

& SetRank
& 0.7154 & 0.5720 & 0.6933 & 0.6132
& 0.8014 & 0.6635 & 0.8145 & 0.7156
& 0.1823 & 0.0896 & 0.2344 & 0.1225 \\

& PRM
& 0.7081 & 0.5639 & 0.6843 & 0.6049
& 0.7992 & 0.6603 & 0.8107 & 0.7122
& 0.1840 & 0.0905 & 0.2368 & 0.1238 \\

\midrule

\multirow{7}{*}{\shortstack{Generator-\\ Evaluator}}
& PIER
& 0.7146 & 0.5715 & 0.6929 & 0.6128
& 0.7987 & 0.6613 & 0.8118 & 0.7130
& \underline{0.1910} & \underline{0.0935}
& \underline{0.2431} & \underline{0.1277} \\

& NAR4Rec
& 0.7348 & 0.5912 & 0.7162 & 0.6338
& 0.8188 & 0.6807 & 0.8365 & 0.7341
& 0.1792 & 0.0880 & 0.2297 & 0.1203 \\

& JDRec
& \underline{0.7399} & \underline{0.5972}
& \underline{0.7233} & \underline{0.6402}
& \underline{0.8255} & \underline{0.6832}
& \underline{0.8409} & \underline{0.7374}
& 0.1832 & 0.0898 & 0.2345 & 0.1227 \\

& OMGRec
& 0.7319 & 0.5886 & 0.7131 & 0.6310
& 0.8040 & 0.6642 & 0.8145 & 0.7160
& 0.1866 & 0.0913 & 0.2385 & 0.1247 \\

& \textbf{DIRECTOR-CVAE}
& \textbf{0.7672} & \textbf{0.6214}
& \textbf{0.7508} & \textbf{0.6641}
& 0.8478 & \textbf{0.7075}
& \textbf{0.8712} & \textbf{0.7588}
& \textbf{0.1979} & \textbf{0.0957}
& 0.2477 & \textbf{0.1305} \\

& \textbf{DIRECTOR-DIFF}
& 0.7659 & 0.6200
& 0.7492 & 0.6628
& \textbf{0.8486} & 0.7069
& 0.8708 & 0.7585
& 0.1976 & 0.0956
& \textbf{0.2480} & 0.1304 \\

& \textbf{Improv.}
& \textbf{+3.69\%} & \textbf{+4.05\%}
& \textbf{+3.80\%} & \textbf{+3.73\%}
& \textbf{+2.80\%} & \textbf{+3.56\%}
& \textbf{+3.60\%} & \textbf{+2.90\%}
& \textbf{+3.61\%} & \textbf{+2.35\%}
& \textbf{+2.02\%} & \textbf{+2.19\%} \\

\bottomrule
\end{tabular}
}
\end{table*}
\section{Experiments}
\label{sec:experiments}

We evaluate \textbf{DIRECTOR} through comprehensive offline experiments and a large-scale online deployment. The offline evaluation is conducted on two public recommendation benchmarks and one industrial full-flow dataset, covering both conventional user-item interaction data and request-level logs collected from a real multi-stage recommendation system. We further deploy \textbf{DIRECTOR} in \textbf{Kuaishou}'s production short-video recommendation platform to examine its effectiveness and efficiency under strict serving constraints. Our experiments are designed to answer the following research questions:

\textbf{RQ1:} How does \textbf{DIRECTOR} compare with reranking baselines?

\textbf{RQ2:} How do transport-guided learning and reward-guided credit assignment contribute to the overall performance?

\textbf{RQ3:} Does it improve multi-objective recommendation when deployed in a large-scale, latency-sensitive production system?

\textbf{RQ4:} Can it reduce online serving resources under matched
throughput, latency, and availability constraints?

\subsection{Experimental Setup}
\label{sec:experiments_setup}

\subsubsection{Datasets}
We conduct offline experiments on two public recommendation datasets, \textbf{ML-1M}~\cite{harper2016movielens} and \textbf{Amazon-Books}~\cite{he2016ups}, and one industrial multi-stage dataset, \textbf{RecFlow}~\cite{liu2024recflow}. For the public datasets, we train \textbf{BPR-MF}~\citep{rendle2009bpr} to simulate upstream retrieval and construct request-specific candidate pools of 50 items. For each user, the six most recent interactions are reserved for testing, and the target slate length is fixed to six. \textbf{RecFlow} directly provides request-level candidates and multi-stage features; we retain 120 upstream candidates per request and generate a slate of six items. Further details are provided in Table~\ref{tab:dataset_stats} and Appendix~\ref{appdix:datasets}.

We further deploy \textbf{DIRECTOR} in the single-column feed of Kuaishou's main application. \textbf{DIRECTOR} and the production baseline are assigned to two disjoint 10\% traffic buckets and evaluated over seven consecutive days. Both groups share the same upstream candidates, input features, Evaluator, and serving objectives.

\subsubsection{Baselines}
We compare \textbf{DIRECTOR} with representative point-wise rankers (\textbf{DNN} and \textbf{DCN}), context-aware rerankers (\textbf{DLCM}~\cite{ai2018learning}, \textbf{PRM}~\cite{pei2019personalized}, and \textbf{SetRank}~\cite{pang2020setrank}), an \textbf{AR} generator (\textbf{Seq2Slate}~\cite{bello2018seq2slate}), and Generator--Evaluator methods (\textbf{PIER}~\cite{shi2023pier}, \textbf{NAR4Rec}~\cite{ren2024nar4rec}, \textbf{JDRec}~\cite{zhao2024jdrec}, and \textbf{OMGRec}~\cite{xu2026omgrec}). We evaluate both \textbf{DIRECTOR-CVAE} and \textbf{DIRECTOR-DIFF}, which instantiate the dynamic-index generator with CVAE and conditional diffusion, respectively. All Generator-Evaluator methods use the same candidate pools, proposal budget, and fixed Evaluator. Detailed baseline descriptions are deferred to Appendix~\ref{appdix:baselines}.

\subsubsection{Metrics.}
We report \textbf{NDCG@6}, \textbf{Precision@6}, \textbf{Recall@6}, and \textbf{F1@6} for recommendation effectiveness. For the online A/B test, we report metrics related to user feedback and serving efficiency. Definitions are given in Appendix~\ref{appdix:metrics}.

\subsubsection{Implementation.}
All methods are optimized with Adam~\cite{kingma2015adam}. Unless otherwise stated, the embedding dimension is 64, the learning rate is $10^{-3}$, and the batch size is 2,048. \textbf{DIRECTOR} is first warm-started with the index-generation and matching objectives and then optimized using the complete reward-guided objective. More reproducibility details are provided in Appendix~\ref{appdix:implementation}.


\begin{table}[t]
\centering
\caption{Online A/B testing and serving stress-test results on Kuaishou APP. 
The online VV improvement is statistically significant ($p<0.05$). 
The efficiency comparison is conducted under the same model QPS and end-to-end latency constraints, while maintaining 99\% service availability.}
\label{tab:online_ab}
\begin{tabular}{llc}
\toprule
\textbf{Category} & \textbf{Metric} & \textbf{Relative Change}\\
\midrule
Effectiveness
 & Valid View & $+0.519\%$ $CI: [0.45\%, 0.59\%]$\\
& Comment & $+0.695\%$ $CI: [0.56\%, 0.83\%]$\\
& Like & $+0.330\%$ $CI: [0.17\%, 0.48\%]$\\
\midrule
Efficiency
 & CPU Consumpt. & $-66.7\%$\\
\bottomrule
\end{tabular}
\end{table}

\subsection{Performance on Recommendation Tasks}
\label{sec:experiments_recommendation}

\textbf{Offline Evaluation (RQ1).}
Table~\ref{tab:main_results} reports the offline performance on
ML-1M, Amazon-Books, and RecFlow. Both \textbf{DIRECTOR} variants consistently outperform all Generator-only and Generator--Evaluator baselines across the three datasets. Compared with the strongest non-DIRECTOR method, the best \textbf{DIRECTOR} variant improves NDCG@6 by 3.69\%, 2.80\%, and 3.61\% on ML-1M, Amazon-Books, and RecFlow, respectively. Consistent improvements are also observed for Precision@6, Recall@6, and F1@6, indicating that \textbf{DIRECTOR} improves both the retrieval of relevant candidates and their ordering within the output slate.

The comparable performance of the two variants shows that \textbf{DIRECTOR} is not tied to
a specific latent generator: both one-pass CVAE generation and
iterative diffusion can produce effective dynamic index matrices.
Meanwhile, the strong performance of DIRECTOR-CVAE suggests that high-quality proposals can already be generated with a single
parallel forward pass, which is attractive for latency-sensitive
deployment. Component-level ablation studies are presented in
Appendix~\ref{appdix:ablation}.

\textbf{Online A/B Test (RQ3 \& RQ4).} 
We further evaluate \textbf{DIRECTOR} through a large-scale online A/B test on Kuaishou's short-video recommendation platform (RQ3), which serves hundreds of millions of users. We hold out 10\% of the online traffic as the control group and assign another 10\% to the treatment group, with the experiment running for seven consecutive days. The control group uses the production \textbf{AR} Generator-Evaluator reranking system, while the treatment group replaces its generator with \textbf{DIRECTOR} and keeps the evaluator and the remaining serving components unchanged. The results show that \textbf{DIRECTOR} achieves a statistically significant relative lift of \textbf{+0.519\%} in VV ($p<0.05$). This improvement demonstrates that the effectiveness of \textbf{DIRECTOR} can be translated into real-world business gains.

We further conduct large-scale stress tests to evaluate the serving efficiency (RQ4) of \textbf{DIRECTOR}. Specifically, we compare \textbf{DIRECTOR} with an NTP-based \textbf{AR} online baseline generator equipped with beam search. To ensure a fair comparison, both methods are evaluated under an identical serving configuration, with the peak throughput controlled at approximately 20,000 QPS and the P99 end-to-end latency constrained to no more than 30\,ms. Meanwhile, both systems are required to maintain 99\% service availability throughout the stress test. As shown in Table~\ref{tab:online_ab}, under strict serving conditions, \textbf{DIRECTOR} reduces machine consumption by \textbf{66.7\%} compared with NTP + Beam Search. This substantial resource saving mainly benefits from \textbf{DIRECTOR}'s parallel generation architecture, which eliminates the sequential decoding dependency inherent in \textbf{AR} beam search while preserving globally coordinated slate construction through direct global hard matching. These results demonstrate that \textbf{DIRECTOR} can achieve comparable peak throughput, latency, and service reliability with significantly fewer computational resources, thereby improving both online recommendation effectiveness and industrial serving efficiency.
\section{Conclusion}
\label{sec:conclusion}

In this work, we propose \textbf{DIRECTOR}, a globally coordinated
non-\textbf{AR} reranking framework for the Generator--Evaluator
paradigm. \textbf{DIRECTOR} jointly generates request-conditioned
dynamic indices for all target positions, uses capacity-constrained
soft transport for conflict-aware training, and performs direct hard
matching for complete and duplicate-free inference without
\textbf{AR} rollouts. We further introduce prefix-anchored credit
assignment to derive position-specific optimization signals from
opaque list-wise rewards.

Our theoretical analysis characterizes the correspondence between
valid slates and assignments, the capacity-induced coupling and
approximation gap of soft transport, and finite-proposal coverage.
Experiments on three offline datasets and Kuaishou's production system demonstrate consistent
improvements in reranking effectiveness and serving efficiency.
These results establish globally coordinated parallel generation as
an effective and deployment-friendly alternative to \textbf{AR}
slate generation in latency-sensitive recommendation systems.

\bibliographystyle{ACM-Reference-Format}
\bibliography{sample-base}

\newpage
\appendix
\appendix

\section{Detailed Proofs of Theoretical Results}
\label{app:theoretical_proofs}

\subsection{Proof of Theorem~\ref{thm:assignment_integrality}}
\label{app:proof_assignment_integrality}

Recall
\[
\mathcal{A}_{n,M}
=
\left\{
\bm A\in\{0,1\}^{n\times M}
\,\middle|\,
\bm A\bm 1_M=\bm 1_n,\;
\bm A^\top\bm 1_n\leq\bm 1_M
\right\},
\]
and let $\mathcal{U}_{n,M}$ be its relaxation obtained by replacing
$\bm A\in\{0,1\}^{n\times M}$ with
$\bm\Gamma\in\mathbb{R}_+^{n\times M}$.
For any slate
$\bm y=(y_1,\ldots,y_n)\in\mathcal{S}(\mathcal C,n)$, define
$A^{\bm y}_{ij}=\mathbb I[y_i=c_j]$. Each row contains exactly one
nonzero entry, and the distinctness of the items in $\bm y$ ensures
that each column contains at most one. Hence
$\bm A^{\bm y}\in\mathcal A_{n,M}$. Conversely, for any $\bm A\in\mathcal A_{n,M}$, each row has a
unique nonzero entry. Let $j_i$ be its column index and set
$y_i=c_{j_i}$. The column constraints imply that the $j_i$'s are
distinct, so $\bm y\in\mathcal S(\mathcal C,n)$. The two mappings are
inverse to each other, establishing the bijection.

Identify each variable $\Gamma_{ij}$ with an edge between position $i$ and candidate $j$ in the complete bipartite graph $K_{n,M}$. The row- and column-sum constraints are described by its node--edge incidence matrix. After changing the signs of the candidate-node rows, this becomes a directed node--arc incidence matrix and is therefore totally unimodular. Replacing equalities by pairs of inequalities and adding nonnegativity constraints preserve total unimodularity. Since the right-hand side is integral, every extreme point of $\mathcal U_{n,M}$ is integral.

Any integral $\bm\Gamma\in\mathcal U_{n,M}$ has nonnegative integer entries and unit row sums; hence every row contains exactly one entry equal to one. The column constraints ensure that no column is selected more than once, so $\bm\Gamma\in\mathcal A_{n,M}$. Therefore,
\[
\operatorname{ext}(\mathcal U_{n,M})
\subseteq
\mathcal A_{n,M}.
\]
Because $\mathcal U_{n,M}$ is a bounded polytope,
\[
\mathcal U_{n,M}
=
\operatorname{conv}\bigl(\operatorname{ext}(\mathcal U_{n,M})\bigr)
\subseteq
\operatorname{conv}(\mathcal A_{n,M}).
\]
The reverse inclusion follows from
$\mathcal A_{n,M}\subseteq\mathcal U_{n,M}$ and the convexity of
$\mathcal U_{n,M}$. Thus,
\[
\mathcal U_{n,M}
=
\operatorname{conv}(\mathcal A_{n,M}).
\]

For $1\leq n\leq M$, $\mathcal U_{n,M}$ is nonempty and compact.
A linear objective therefore attains an optimum at an extreme point
$\bm\Gamma^\star$ of $\mathcal U_{n,M}$. By the preceding result,
$\bm\Gamma^\star\in\mathcal A_{n,M}$. Together with
$\mathcal A_{n,M}\subseteq\mathcal U_{n,M}$, 
\[
\max_{\bm A\in\mathcal A_{n,M}}
\langle\bm A,\bm S\rangle
=
\max_{\bm\Gamma\in\mathcal U_{n,M}}
\langle\bm\Gamma,\bm S\rangle.
\]
Hence the relaxed problem admits an integral optimizer corresponding,
through the above bijection, to a valid duplicate-free slate.
\hfill$\square$

\subsection{Proof of Theorem~\ref{thm:ot_coupling}}
\label{app:proof_ot_coupling}

\subsubsection{Existence and uniqueness.}
Let $F(\bm{\Gamma})=\langle\bm{\Gamma},\bm{S}\rangle+\mu H(\bm{\Gamma})$. The uniform plan $\Gamma^0_{ij}=1/M$ belongs to
$\mathcal U_{n,M}$ because its row sums equal one and its column sums equal $n/M\leq1$. Hence, $\mathcal U_{n,M}$ is nonempty. It is also compact, so the continuity of $F$ guarantees the existence
of an optimizer. Since $-u\log u$ is strictly concave on $\mathbb R_+$,
$H(\bm\Gamma)$ is strictly concave. Therefore, $F$ is strictly
concave on the convex set $\mathcal U_{n,M}$, and its optimizer
$\bm\Gamma_\mu^\star$ is unique.

\subsubsection{Positivity.}
Suppose that $\Gamma_{\mu,ij}^\star=0$ for some $(i,j)$ and define
\[
\bm\Gamma_\delta
=
(1-\delta)\bm\Gamma_\mu^\star+\delta\bm\Gamma^0,
\quad \delta\in(0,1).
\]
This plan remains feasible. The change in the linear objective and
the entropy contributions of positive entries are $O(\delta)$,
whereas each zero entry gains
\[
-\frac{\delta}{M}\log\frac{\delta}{M}
=
\frac{\delta}{M}\log\frac{M}{\delta}
=
\Theta\!\left(\delta\log\frac{1}{\delta}\right).
\]
For sufficiently small $\delta$, this positive entropy gain dominates,
giving $F(\bm\Gamma_\delta)>F(\bm\Gamma_\mu^\star)$, a contradiction.
Thus, $\Gamma_{\mu,ij}^\star>0$ for all $i,j$.

\subsubsection{KKT characterization.}
Introduce multipliers $\alpha_i\in\mathbb R$ for the row equalities
and $\beta_j\geq0$ for the column-capacity constraints. Since
$\bm\Gamma_\mu^\star$ is strictly positive, the nonnegativity
constraints are inactive. The KKT stationarity and complementary
slackness conditions give
\[
S_{ij}
-\mu\bigl(\log\Gamma_{\mu,ij}^\star+1\bigr)
-\alpha_i-\beta_j
=0,
\quad
\beta_j
\left(
1-\sum_{i=1}^{n}\Gamma_{\mu,ij}^\star
\right)
=0.
\]
Therefore,
\[
\Gamma_{\mu,ij}^\star
=
\exp\left(
\frac{S_{ij}-\alpha_i-\beta_j}{\mu}-1
\right).
\]

\subsubsection{Comparison with independent normalization.}
Without the column-capacity constraints, the problem separates across
rows, and its unique optimizer is
\[
\overline{\Gamma}_{ij}
=
\frac{\exp(S_{ij}/\mu)}
{\sum_{\ell=1}^{M}\exp(S_{i\ell}/\mu)}.
\]
If $\overline{\bm\Gamma}\in\mathcal U_{n,M}$, it is feasible for the
constrained problem and already maximizes $F$ over the larger
row-constrained set. Uniqueness therefore gives
$\bm\Gamma_\mu^\star=\overline{\bm\Gamma}$.

If $\overline{\bm\Gamma}\notin\mathcal U_{n,M}$, it cannot equal the
constrained optimizer. Moreover, if all $\beta_j$ were zero, the KKT
conditions would reduce to independent row normalization, implying
$\bm\Gamma_\mu^\star=\overline{\bm\Gamma}$, which is impossible.
Hence, $\beta_j>0$ for at least one candidate $j$. Since the same
$\beta_j$ appears in the optimality condition of every position, it
couples their assignments through the shared candidate capacity.
\hfill$\square$

\subsection{Proof of Theorem~\ref{thm:entropic_approximation}}
\label{app:proof_entropic_approximation}

Let
\[
    \bm A^\star
    \in
    \operatorname*{arg\,max}_{\bm A\in\mathcal A_{n,M}}
    \langle \bm A,\bm S\rangle,
\]
so that
$V_{\mathrm{hard}}(\bm S)=\langle\bm A^\star,\bm S\rangle$.
By Theorem~\ref{thm:assignment_integrality},
$\bm A^\star\in\mathcal U_{n,M}$. Moreover, each row of
$\bm A^\star$ is one-hot, and hence $H(\bm A^\star)=0$.

Since $\bm\Gamma_\mu^\star$ maximizes the entropy-regularized
objective over $\mathcal U_{n,M}$,
\[
    \langle\bm\Gamma_\mu^\star,\bm S\rangle
    +\mu H(\bm\Gamma_\mu^\star)
    \geq
    \langle\bm A^\star,\bm S\rangle
    +\mu H(\bm A^\star) =V_{\mathrm{hard}}(\bm S).
\]
Therefore,
\[
    V_{\mathrm{hard}}(\bm S)
    -
    \langle\bm\Gamma_\mu^\star,\bm S\rangle
    \leq
    \mu H(\bm\Gamma_\mu^\star).
\]

Each row of $\bm\Gamma_\mu^\star$ is a probability vector over
$M$ candidates, so its entropy is at most $\log M$. Summing over
the $n$ rows gives
\[
    H(\bm\Gamma_\mu^\star)\leq n\log M.
\]
Hence,
\[
    V_{\mathrm{hard}}(\bm S)
    -
    \langle\bm\Gamma_\mu^\star,\bm S\rangle
    \leq
    \mu n\log M.
\]

Finally, since $\bm\Gamma_\mu^\star\in\mathcal U_{n,M}$ and
$V_{\mathrm{hard}}(\bm S)$ is also the maximum linear score over
$\mathcal U_{n,M}$ by
Theorem~\ref{thm:assignment_integrality},
\[
    \langle\bm\Gamma_\mu^\star,\bm S\rangle
    \leq
    V_{\mathrm{hard}}(\bm S).
\]
This proves both bounds.
\hfill$\square$

\subsection{Local Stability of Hard Matching}
\label{app:local_decode_stability}

\begin{proposition}[Local Decoding Stability]
\label{prop:local_decode_stability}
Suppose $\bm y$ is the unique output of the hard decoder at
$\bm Q_0$, with a positive assignment margin over all competing
slates. Then there exists an open neighborhood
$\mathcal N(\bm Q_0)$ such that
\[
    \mathsf D(\bm Q,\bm E)=\bm y,
    \quad
    \forall\,\bm Q\in\mathcal N(\bm Q_0).
\]
If
$p_\theta(\mathcal N(\bm Q_0)\mid\bm h,\bm P)>0$,
then DIRECTOR generates $\bm y$ with positive probability.
\end{proposition}

\begin{proof}
For each feasible assignment
$\bm A\in\mathcal A_{n,M}$, define its hard-matching score as
\[
    g_{\bm A}(\bm Q)
    =
    \langle\bm A,\bm S(\bm Q)\rangle,
    \quad
    \bm S(\bm Q)=\bm Q\bm E^\top/\tau.
\]
Since $\bm S(\bm Q)$ is continuous in $\bm Q$, every
$g_{\bm A}$ is continuous.
Let $\bm A_{\bm y}$ denote the assignment corresponding to
$\bm y$. By assumption,
\[
    g_{\bm A_{\bm y}}(\bm Q_0)
    >
    \max_{\bm A\neq\bm A_{\bm y}}
    g_{\bm A}(\bm Q_0).
\]
Because $\mathcal A_{n,M}$ is finite and the margin is positive,
the same strict inequality holds in some open neighborhood
$\mathcal N(\bm Q_0)$. Hence, $\bm A_{\bm y}$ remains the unique
optimal assignment throughout this neighborhood, so every
$\bm Q\in\mathcal N(\bm Q_0)$ is decoded to $\bm y$. Therefore,
\[
    \mathbb P\!\left[
        \mathsf D(\bm Q,\bm E)=\bm y
    \right]
    \geq
    p_\theta(
        \mathcal N(\bm Q_0)\mid\bm h,\bm P
    )
    >0.
\]
\end{proof}

\subsection{Proof of Proposition~\ref{prop:multi_sample_discovery}}
\label{app:proof_multi_sample_discovery}

Each independently generated proposal belongs to
$\mathcal T$ with probability $p_{\mathcal T}$. Hence, the probability
that none of the $K$ proposals belongs to $\mathcal T$ is
$(1-p_{\mathcal T})^K$. Taking the complement gives
\[
    \mathbb P
    \left[
        \exists k\in[K]:
        \bm y^{(k)}\in\mathcal T
    \right]
    =
    1-(1-p_{\mathcal T})^K.
\]
This probability is nondecreasing in $K$ and converges to one whenever
$p_{\mathcal T}>0$.
\hfill$\square$
\section{Experimental Details}
\label{appdix:experiments}

\subsection{Datasets and Preprocessing}
\label{appdix:datasets}

For every request, the reranker selects an ordered slate of six
distinct items from a request-specific candidate pool. All compared
methods use identical data splits, candidate pools, input features,
and evaluation labels.

\textbf{ML-1M}~\citep{harper2016movielens} and
\textbf{Amazon-Books}~\citep{he2016ups} contain timestamped user-item interactions but do not provide request-level candidate pools. We first apply iterative 20-core filtering, retaining only users and
items with at least 20 interactions. The remaining interactions of each user are sorted chronologically and partitioned into
non-overlapping lists of length six. Users with fewer than three
complete lists are discarded.

For each retained user, the last list is used for testing, the
penultimate list is used for validation, and all preceding lists are
used for training. For every target list, its historical input contains
all interactions strictly preceding that list; neither the target
items nor any future interactions are included in the history. This
chronological construction prevents information leakage across
training, validation, and test instances.

BPR-MF~\citep{rendle2009bpr} is trained using only interactions from
the training portion and tuned on the validation portion to emulate
the upstream retrieval stage. For each reranking instance, the six
target items are combined with highly ranked BPR-MF items. Additional
candidates are drawn from the top-200 retrieval results after
duplicate removal until a pool of 50 distinct items is obtained. The
retriever, histories, candidate pools, and data partitions are fixed
across all compared methods.

\textbf{RecFlow}~\citep{liu2024recflow} provides request-level logs collected
from multiple stages of an industrial recommendation pipeline. We
follow the official setting and use its second period, from February
5 to February 18. For each request, we retain the top 120 items
entering the reranking stage and preserve their upstream ranking
positions as input features. The model inputs include video
identifiers, category- and author-related attributes, and the user's
50 most recent interactions.

Feedback labels are observed only for items recorded in
\texttt{realshow}. An exposed item is labeled positive if it receives
effective-view feedback, while all remaining candidates, including
unexposed items, are assigned a negative label. This conservative
labeling protocol introduces severe class imbalance and partially
explains the lower absolute metric values on RecFlow. Although it
does not eliminate exposure bias, it is applied identically to all
methods using the same requests, candidate pools, and labels, thereby
supporting a controlled relative comparison.

\subsection{Baseline Descriptions}
\label{appdix:baselines}

We compare \textbf{DIRECTOR} with representative methods covering
point-wise ranking, context-aware reranking, autoregressive generation,
and Generator--Evaluator reranking.

\textbf{Generator-Only Methods}

\textbf{DNN} is a standard point-wise ranking model that concatenates
user, item, and request features and predicts an independent relevance
score for each candidate through a multilayer perceptron. The final
slate is obtained by sorting candidates according to their predicted
scores.

\textbf{DCN} extends the point-wise DNN baseline with explicit
feature-crossing layers. It models interactions among user, item, and
contextual features, while still scoring each candidate independently
before sorting.

\textbf{Seq2Slate}~\cite{bello2018seq2slate} formulates reranking as
sequential slate generation. A pointer-network decoder selects one
item at each position conditioned on the request context, candidate
representations, and the previously generated prefix. Selected items
are masked to prevent duplicate outputs.

\textbf{DLCM}~\cite{ai2018learning} applies a recurrent contextual
model to the initial candidate ranking. Candidate representations are
updated using information from the surrounding list, after which
refined scores are used to construct the final slate.

\textbf{PRM}~\cite{pei2019personalized} employs self-attention to
model mutual influences among candidate items and incorporates
personalized user information into reranking. The final slate is
obtained by sorting the resulting context-aware scores.

\textbf{SetRank}~\cite{pang2020setrank} uses permutation-invariant set
modeling to encode the candidate pool. It captures global candidate
interactions without relying on the input order and produces a
contextualized score for each item.

\textbf{Generator-Evaluator Methods}

\textbf{PIER}~\cite{shi2023pier} adopts a Generator--Evaluator
architecture for permutation-level reranking. Its Generator produces
multiple candidate permutations, while its Evaluator estimates their
list-wise utilities and selects the final output.

\textbf{NAR4Rec}~\cite{ren2024nar4rec} predicts all target positions
in parallel. It introduces matching and sequence-level objectives to
reduce inconsistent or duplicate position-wise predictions, together
with an additional decoding procedure for constructing a valid slate.

\textbf{JDRec}~\cite{zhao2024jdrec} is a practical actor--critic
framework for combinatorial recommendation. Its actor learns a
slate-generation policy, while the critic provides reward guidance
for policy optimization and deployment-oriented training.

\textbf{OMGRec}~\cite{xu2026omgrec} directly constructs a
request-conditioned position--candidate allocation matrix and applies
one-time matching to obtain a valid permutation. It further introduces
permutation-level representations to model the overall slate structure.

\textbf{DIRECTOR-CVAE} instantiates the dynamic-index generator with a
conditional variational autoencoder. During training, the posterior is
conditioned on the target index matrix and request representation,
whereas the conditional prior is used to sample complete index matrices
at inference. All target positions are generated jointly in a single
forward pass.

\textbf{DIRECTOR-DIFF} uses a conditional diffusion model to generate
the dynamic-index matrix. Starting from Gaussian noise, the reverse
process jointly denoises all position indices at each step. Both
variants use capacity-constrained soft transport for conflict-aware
training and directly apply global hard matching to the similarity
matrix at inference.

\subsection{Main Comparison Protocol}
\label{appdix:main_protocol}

To ensure a fair comparison, all methods use identical data splits, user histories, candidate pools, input features, and candidate ordering. Each method outputs an ordered slate of six distinct items. Baselines retain their original duplicate-resolution strategies, whereas \textbf{DIRECTOR} enforces item exclusivity through global hard matching.

All Generator--Evaluator methods use the same proposal budget $K=20$. For each baseline, the proposals are produced using the sampling strategy described in its original paper and implementation, thereby preserving its native stochastic generation procedure and decoding mechanism. We standardize only the number of sampled
proposals and their final rescoring: the same frozen Evaluator selects the highest-scoring slate from the 20 proposals. Generator-only methods directly output a single final slate.

\subsection{Evaluation Metrics}
\label{appdix:metrics}

Let $\bm y=(y_1,\ldots,y_6)$ denote the output slate,
$\operatorname{rel}(y_i)\in\{0,1\}$ the relevance label of item
$y_i$, and $\mathcal C^{+}$ the set of relevant items in the
candidate pool. We compute
\[
\begin{aligned}
&\operatorname{DCG@6}
=
\sum_{i=1}^{6}
\frac{2^{\operatorname{rel}(y_i)}-1}{\log_2(i+1)},\ \operatorname{NDCG@6}
=
\frac{\operatorname{DCG@6}}{\operatorname{IDCG@6}},\\
&\operatorname{Precision@6}
=
\frac{1}{6}\sum_{i=1}^{6}\operatorname{rel}(y_i),\ 
\operatorname{Recall@6}=
\frac{\sum_{i=1}^{6}\operatorname{rel}(y_i)}
{|\mathcal C^{+}|},\\
&\operatorname{F1@6}
=
\frac{
2\cdot\operatorname{Precision@6}\cdot\operatorname{Recall@6}
}{
\operatorname{Precision@6}+\operatorname{Recall@6}
}.
\end{aligned}
\]
Requests with $|\mathcal C^{+}|=0$ are excluded from recall-based
metrics, and F1@6 is set to zero when both precision and recall are
zero. All metrics are computed per request and then averaged over
the test set.

\subsection{Implementation Details}
\label{appdix:implementation}

All offline tasks are implemented in PyTorch and optimized with
Adam~\cite{kingma2015adam}. Unless otherwise specified, the embedding dimension is 64, the batch size is 2,048, the initial learning rate is $10^{-3}$, and the target slate length is six. Candidate embeddings and request representations are computed once per request and shared across all proposals. The Evaluator remains frozen during training.

The proposal budget is identical for all Generator-Evaluator methods. The temperature $\tau$, entropy coefficient $\mu$, loss weights $\alpha$ and $\lambda$, latent dimension, and generator-specific parameters are selected on the validation set. For \textbf{DIRECTOR-DIFF}, this also includes the diffusion schedule and number of reverse steps. 

The controlled offline experiments use a lightweight point-wise relevance model. Given the request context and candidate features, it predicts an item-level relevance score $r_{\phi}(y_i\mid\bm x,\mathcal C)$ and is trained separately using the standard binary cross-entropy loss. The scalar reward of a slate is obtained by directly summing its item-level scores:
\[
R_\phi(\bm y\mid\bm x,\mathcal C)
=
\sum_{i=1}^{n}
r_\phi(y_i,i\mid\bm x,\mathcal C).
\]
The Evaluator is pretrained before Generator optimization and remains frozen throughout the training of all DIRECTOR variants. The same frozen Evaluator is used for reward-guided optimization and final proposal selection.

All methods use the same candidate pools and evaluation scripts. Each offline experiment is repeated with five random seeds, and the reported results are averaged over the five runs. The code will be released upon acceptance.

\begin{table}[t]
\centering
\caption{Ablation study of \textbf{DIRECTOR-CVAE} on RecFlow.
Global hard matching is retained in all variants.}
\label{tab}
\setlength{\tabcolsep}{4.5pt}
\renewcommand{\arraystretch}{1.08}
\begin{tabular}{lcccc}
\toprule
Model & N@6 & P@6 & R@6 & F1@6 \\
\midrule
w/o Transport Guidance
& 0.1675 & 0.0812 & 0.2110 & 0.1114 \\

w/o Reward Guidance
& 0.1847 & 0.0903 & 0.2342 & 0.1231 \\

w/ Global Credit
& 0.1894 & 0.0921 & 0.2395 & 0.1253 \\

\textbf{DIRECTOR-CVAE}
& \textbf{0.1979} & \textbf{0.0957}
& \textbf{0.2477} & \textbf{0.1305} \\
\bottomrule
\end{tabular}
\end{table}

\subsection{Ablation Study}
\label{appdix:ablation}

Table~\ref{tab} evaluates the main training components of
\textbf{DIRECTOR-CVAE}. All variants retain the same dynamic-index
generator and global hard matching at inference; therefore, every
output remains complete and duplicate-free.

\textbf{w/o Transport Guidance} replaces the capacity-constrained transport plan with independent row-wise normalization in the matching and credit-assignment losses. Its substantial degradation shows that modeling competition over shared candidate capacity provides important cross-position supervision. \textbf{w/o Reward Guidance} removes $\mathcal L_{\mathrm{CA}}$ and optimizes only index generation and supervised matching, confirming the benefit of aligning the Generator with list-wise utility. \textbf{w/ Global Credit} assigns the same slate-level advantage $R_\phi(\bm y)-R_\phi(\bm b)$ to every position instead of using the prefix-anchored credits $\{\Delta_i\}_{i=1}^{n}$. Its lower performance demonstrates that position-specific credit assignment provides more informative optimization signals. The complete model performs best across all metrics, validating the complementary effects of transport-guided learning and prefix-anchored credit assignment.

\subsection{Online A/B Testing Protocol}
\label{appdix:online_protocol}

We deploy \textbf{DIRECTOR} in the reranking stage of the
single-column feed on Kuaishou's main application. The platform
supports multi-generator serving, allowing proposals from multiple
generation routes to be jointly recalled and evaluated.
\textbf{DIRECTOR-CVAE} and \textbf{DIRECTOR-DIFF} are therefore both
deployed as parallel Generator channels and contribute proposals to
the same downstream Evaluator.

The production baseline and the joint DIRECTOR deployment are assigned
to two mutually exclusive traffic buckets, each receiving 10\% of the
total traffic, and are evaluated for seven consecutive days. Both
groups share the same upstream modules, candidate pools, input
features, Evaluator, and multi-objective serving targets. Accordingly,
the reported effectiveness results represent the aggregate lift of
deploying the two DIRECTOR variants together, while the reported
serving-efficiency result is averaged across the two deployed
generation routes rather than attributed to either variant
individually.

Due to confidentiality, we cannot disclose the complete
architecture, beam configuration, absolute machine count, or detailed resource-accounting procedure of the production baselines. We therefore report only normalized relative resource consumption under
strictly matched serving constraints, including approximately 20,000 QPS, P99 end-to-end latency below 30\,ms, and service availability of at least 99\%. The reported CPU reduction measures the relative resources required by the two systems to satisfy these same production constraints.

\end{document}